\documentclass[11pt, oneside]{article}   

\usepackage {geometry}               	
\usepackage{graphicx}				
\usepackage{tikz}
\usepackage{titling}
\usepackage{float}
\usepackage[pdftex,colorlinks=true,linkcolor=blue,citecolor=blue,urlcolor=black]{hyperref}						
\usepackage{thmtools,thm-restate}								
\usepackage{amssymb}
\usepackage{amsmath}
\usepackage{amsthm}
\usepackage{braket}
\usepackage{xcolor}

\newtheorem{theorem}{Theorem}[section]
\newtheorem{lemma}[theorem]{Lemma}
\newtheorem{corollary}[theorem]{Corollary}
\newtheorem{definition}[theorem]{Definition}

\newtheorem{observation}[theorem]{Observation}
\newtheorem{conjecture}[theorem]{Conjecture}

\DeclareMathOperator{\poly}{poly}

\newcommand{\ketbra}[2]{\left| #1 \rangle\langle #2\right|}
\newcommand{\proj}[1]{\ketbra{#1}{#1}}

\newcommand{\bit}{\{0,1\}}

\newcommand{\BPP}{\mathsf{BPP}}

\newcommand{\BQP}{\mathsf{BQP}}
\newcommand{\PBQP}{\mathsf{PostBQP}}

\newcommand{\DQC}{\mathsf{DQC1}}
\newcommand{\PDQC}{\mathsf{PostDQC1}}

\newcommand{\sep}{\mathsf{sep}}
\newcommand{\PostBPP}{\mathsf{PostBPP}}

\usepackage{array}

\title{The one clean qubit model without entanglement is classically simulable}
\author{Mithuna Yoganathan\thanks{Email: \href{mailto:my332@cam.ac.uk}{my332@cam.ac.uk}}\\
University of Cambridge, UK  \and Chris Cade\thanks{Email: \href{mailto:chris.cade@bristol.ac.uk}{chris.cade@bristol.ac.uk}} \\ University of Bristol, UK}
\date{}	

\begin{document}
\maketitle

\newcommand{\CU}[3]{U^{\mathcal{#1}}_{#2, #3}}
\newcommand{\UC}[3]{U^{#2, #3}_{\mathcal{#1}}}
\newcommand{\ketp}[1]{\ket{{#1}^\perp}}
\newcommand{\unitary}[4]{\begin{pmatrix} #1 & #2 \\ #3 & #4 \end{pmatrix}}


\begin{abstract}
Entanglement has been shown to be necessary for pure state quantum computation to have an advantage over classical computation. However, it remains open whether entanglement is necessary for quantum computers that use mixed states to also have an advantage. The one clean qubit model is a form of quantum computer in which the input is the maximally mixed state plus one pure qubit. Previous work has shown that there is a limited amount of entanglement present in these computations, despite the fact that they can efficiently solve some problems that are seemingly hard to solve classically. This casts doubt on the notion that entanglement is necessary for quantum speedups. In this work we show that entanglement is indeed crucial for efficient computation in this model, because without it the one clean qubit model is efficiently classically simulable. 
\end{abstract}

\section{Introduction}

Entanglement is often conjectured to be the source of the speedups achieved by quantum computers. This claim is supported by the results of Jozsa and Linden \cite{jozsa2003role} and Vidal \cite{vidal2003efficient}, which showed that pure state computations with only small amounts of entanglement can be efficiently classically simulated. Therefore entanglement is clearly necessary for pure state quantum computation. However, it has not been shown that separable mixed state computations are classically simulable. Deciding this is one of the questions posed in the \textit{``Ten Semi-Grand Challenges for Quantum Computing Theory''} raised by Aaronson in 2005 \cite{aaronson2005ten}. 

If mixed state quantum computers can obtain some advantage without entanglement, this would imply that entanglement is not the only source of quantum advantage. Such a result would be plausible even though it does not hold in the pure case. Pure states with restricted entanglement can be described very efficiently, while for mixed states without entanglement this is not obviously the case. 

One possible reason that entanglement may be necessary for pure, but not mixed, states is that entanglement is the only type of correlation present in the former, but not in the latter. Perhaps it is these other correlations that are important in a quantum computer, and not just entanglement itself. Indeed, Vidal's \cite{vidal2003efficient} algorithm can only efficiently simulate mixed state computations when the total correlations (entanglement as well as classical correlations) are restricted. 

Another argument against entanglement being the only resource responsible for quantum advantages arises from studying the One Clean Qubit model, whose corresponding complexity class is known as $\DQC$ (Deterministic quantum computation with one clean qubit). This model is one of the few known nontrival mixed state quantum computers. The input state to this computer is one pure (or `clean') qubit and $n$ qubits in the maximally mixed state \cite{knill1998power}. Any polynomial-sized quantum circuit can be applied to these qubits, after which the initially clean qubit is measured. Despite how unresourceful the initial state appears to be, $\DQC$ can perform several tasks seemingly exponentially faster than is classically possible, such as the following: computing the normalised trace of a unitary \cite{shor2007estimating}; estimating a coefficient in the Pauli decomposition of a quantum circuit up to polynomial accuracy \cite{knill1998power}; computing Schatten $p$-norms up to a suitable level of accuracy \cite{cade2017quantum}; and computing the trace closure of a Jones polynomial\cite{shor2007estimating}. The power of the class remains the same even if we allow more than one clean qubit \cite{shor2007estimating} -- that is, $\DQC = \mathsf{DQCk}$ for $k = O(\log(n))$ -- suggesting that the initial state is indeed more resourceful than it first appears to be. Furthermore, several results have shown that $\DQC$ cannot be classically simulated under some complexity theoretic conjectures \cite{fujii2015power, morimae2014hardness, morimae2017hardness}. 

This is all despite the fact that $\DQC$ does not require entanglement to be present between the clean qubit and the maximally mixed register in order to demonstrate an advantage over classical computation \cite{poulin2004exponential}. Moreover, the amount of entanglement in the computer, as measured by the multiplicative negativity, is always bounded by a constant independent of the number of qubits \cite{datta2005entanglement}. $\DQC$ can still have strong correlations besides entanglement \cite{datta2007role}, however, meaning that it cannot be simulated by the algorithm in Ref. \cite{vidal2003efficient}. For these reasons, entanglement was postulated to not be vital for computational speedups achieved by the one clean qubit model \cite{datta2007role}. 

In this work we resolve this question by showing that entanglement is in fact necessary in $\DQC$. Any circuit from $\DQC$ that does not generate entanglement can be efficiently classically simulated. We prove this by characterising these circuits. We consider two ways of enforcing that the circuit produce no entanglement; requiring the states to be seperable after each gate applied, or also requiring the state to be seperable throughout the entire computation. Even though circuits produced in either case are surprisingly nontrivial, we show that a classical simulation of them can be performed.

Our result shows that despite the limited amount of it present in $\DQC$, entanglement is playing a crucial role in the quantum speedups achieved by the one clean qubit model. This suggests that entanglement may be necessary in other mixed state quantum computers after all.


\section{The one clean qubit model}\label{sec:DQC1_prelims}
The \textbf{one clean qubit model} of computation has input state $|+\rangle\langle +|\otimes \mathbb{I}/2^n$, where $|+\rangle=(|0\rangle +|1\rangle)/\sqrt{2}$. Then the clean qubit is used to apply a controlled $n$ qubit unitary $U$ on the mixed register. Note that $U$ must be constructed from a polynomial number of constant size gates\footnote{If $U= U_m...U_1$, then control $U$ can be constructed by applying control $U_1$, control $U_2$ etc.}. Finally, the the originally clean qubit is measured in either the Pauli $X$ or $Y$ basis. It is also possible to define the one clean qubit model so as to allow a circuit to be applied to all the qubits. It can be shown that both definitions give rise to the same complexity class (see e.g. Ref \cite{shor2007estimating} for a constructive proof of this fact) and so we will restrict our attention to the above formulation throughout this work. 
To be explicit: when we refer to the one clean qubit model, we mean a circuit of the form $\proj{0}\otimes \mathbb{I} + \proj{1}\otimes U$ is applied to the state $|+\rangle\langle +|\otimes \mathbb{I}/2^n$, and the first qubit is measured in the $X$ or $Y$ basis. 

To understand why this model can compute the normalised trace, let us consider the decomposition of the maximally mixed state 
\begin{equation}
\mathbb{I}/2^n= \sum_i \frac{1}{2^n} \proj{u_i},
\end{equation}
where $\{|u_i\rangle\}_i$ is an orthonormal basis formed of eigenvectors of $U$.

For convenience we will use the notation $\{p(i), |\psi_i\rangle\}_i$, where $p$ is a probability distribution, to denote the mixed state $\sum_i p(i) |\psi_i\rangle\langle \psi_i|$. This notation highlights that the state can be thought of as a probabilistic mixture of pure states, but that the ensemble is generally not unique. 

In this case the above equation implies that the initial state can be considered to be $\{\frac{1}{2^n}, |+\rangle |u_i\rangle\}_i$. Let $\lambda_i$ be the eigenvalue of $|u_i\rangle$. Under the action of controlled $U$ the state goes to
\begin{equation} \label{eq:final}
|+\rangle |u_i\rangle \rightarrow \frac{|0\rangle|u_i\rangle+\lambda|1\rangle|u_i\rangle}{\sqrt{2}}= \frac{|0\rangle + \lambda_i|1\rangle}{\sqrt{2}}|u_i\rangle.
\end{equation}

If this (pure) state was measured in the $X$ (or $Y$) basis the expectation value would be $\textrm{Re}(\lambda_i)$ (or $\textrm{Im}(\lambda_i)$). Because we must average over all eigenvectors in the basis, the actual quantity that is estimated this way is $\sum_i \lambda_i /2^n=\textrm{Tr}(U)/2^n$. This quantity is called the normalised trace, and the above method estimates it to inverse polynomial additive error. This does not allow us to compute the trace itself very accurately, but nevertheless appears to be a difficult quantity to compute classically.

We can use this model to define the class of decision problems \textbf{$\DQC$} -- the class of decision problems that can be decided correctly with probability $1/2+\epsilon$ using the one clean qubit model, where $\epsilon$ is at most inverse polynomially small. 

The notion of \emph{efficient classical simulation} (and its shorthand, \emph{classically simulable}) that we use in this work is the following: for some uniform family of quantum circuits $\mathcal{F}_n$ acting on the $n+1$-qubit state $\rho := \proj{+} \otimes \frac{\mathbb{I}}{2^n}$, we say that the family $\mathcal{F}_n$ can be efficiently classically simulated if, for any circuit from $\mathcal{F}_n$, we can estimate the probability $p_{X}(1)$ (or $p_{Y}(1)$) of obtaining outcome $1$ when measuring the on the clean qubit in the $X$ (or $Y$) basis at the end of the circuit up to additive error $1/O(\poly(n))$ in time $O(\poly(n))$.

\section{The one clean qubit model without entanglement}

\begin{definition} \textbf{(Separable mixed state)}
A mixed state is separable on the partition $A|B$ if it is described by an ensemble in this form: $\{p(i), |\psi^i\rangle_A \otimes |\phi^i\rangle_B\}_i$. 
In words, this means the mixed state is separable across $A|B$ if it can be written as a probabilistic mixture of pure states separable across that cut. 
\end{definition}

The following theorem gives our first constraint on the entanglement (or lack thereof) present in the one clean qubit model.
\begin{theorem} \label{thm:noent} \textbf{(From Ref \cite{poulin2004exponential})}
The one clean qubit computations do not have entanglement between the `clean' qubit and the `noisy' register at any point in the computation. 
\end{theorem}

\begin{proof}
The pure state in equation \ref{eq:final} has no entanglement across this partition. The final state of a one clean qubit computation is a (uniform) probabilistic mixture of these pure states. Hence the state is separable.
\end{proof}

Despite there being no entanglement across this bipartition , there clearly are correlations between the two registers. Also note that there can still be entanglement across other cuts in the one clean qubit model. Therefore we will define the \emph{one clean qubit model without entanglement} as being the one clean qubit model (in the sense described in Section \ref{sec:DQC1_prelims}) with no entanglement across \emph{any} cut during the computation. We will define \textbf{$\DQC_{\sep}$} to be the class of decision problems that can be efficiently decided in this model. Note that the separability condition can be enforced two ways. If the circuit is composed of local gates, separability can be required after each gate is applied. Alternatively, if the gates are applied in a continuous manner, say by applying a Hamiltonian evolution, it is natural to enforce separability at all points in time. In this work we consider both possibilities. Theorem \ref{thm:DQC1sep} refers to the discrete gate case and Theorem \ref{thm:DQC1sepcont} to the continuous case. 

Here we show that the final state of a one clean qubit model computation (in which a unitary $U$ is applied to the mixed register, controlled on the clean qubit) is separable if and only if the unitary $U$ satisfies a particular condition; namely that it must have a separable eigenbasis:

\begin{definition}
\textbf{($U$ has a separable eigenbasis)} We say that an $n$ qubit unitary $U$ has a separable eigenbasis if it is possible to write $U=\sum_i \lambda_i |u_i\rangle\langle u_i|$, where each $|u_i\rangle$ is separable. 
\end{definition}
Note that this does not require all eigenvectors to be separable. If $U$ has degenerate eigenvalues then there are infinitely many eigenbases for $U$. However, only one need be separable to satisfy the above condition.

\begin{theorem} \label{thm:sep}
The final state of a one clean qubit model computation (in which the circuit applied to the qubits is of the form $\proj{0}\otimes \mathbb{I} + \proj{0}\otimes U$) has no entanglement if and only if $U$ has a separable eigenbasis. 
\end{theorem}

We prove this theorem in Section \ref{sec:main_proofs}.

\begin{theorem} \label{thm:DQC1sep}
Let $\mathcal{C}$ be a circuit in $\textrm{DQC1}_\textrm{sep}$. In this circuit control $U$ is applied, where $U=U_r...U_1$. Then the following unitaries must have a product eigenbasis: $U_1$, $U_2U_1, U_3U_2U_1$,..., $U_r...U_1$.
\end{theorem}

\begin{theorem} \label{thm:DQC1sepcont}
Let $\mathcal{C}$ be a circuit in $\textrm{DQC1}_\textrm{sep}$. Suppose the controlled gate applied after time $t$ is $U(t)$, $U(0)=\mathbb{I}$ and at time $T$ the full gate is applied, $U(T)=U$. Then for all $0\leq t\leq T$, $U(t)$ must have a product eigenbasis. 
\end{theorem}

\begin{proof}
The theorems follow from enforcing the separability condition after each gate (or at every point in time in the continuous version), and applying Theorem \ref{thm:sep}. 
\end{proof}


\section{$\textrm{DQC1}_\textrm{sep}$ circuits}

Though Theorems \ref{thm:DQC1sep} and \ref{thm:DQC1sepcont} characterise the circuits that make up $\DQC_{\sep}$, this characterisation is not explicit. In this section we demonstrate which $1$ and $2$ qubit gates can be used to construct the circuits. We will start with an illustrative example of a $2$ qubit gate that has a product eigenbasis. 

\begin{lemma}\label{flos} 
Gates in the following form have a product eigenbasis:
\begin{equation}
\left[\begin{array}{cc}
B & 0\\
0 & C
\end{array}\right],
\end{equation}
where $B$ and $C$ are $1$ qubit unitaries. This gate applies $B$ to the second qubit if the first qubit is in state $|0\rangle$ and $C$ if the first qubit's state is $|1\rangle$.
\end{lemma}

\begin{proof}
The product eigenbasis of this unitary is $\{|0\rangle|b\rangle, |0\rangle|b^{\perp}\rangle,|1\rangle|c\rangle, |1\rangle|c^{\perp}\rangle \}$, where $|b\rangle$ and $|b^{\perp}\rangle$ are eigenvectors of $B$ and $|c\rangle$ and $|c^{\perp}\rangle$ are eigenvectors of $C$. 
\end{proof}

We will generalise this example so that the unitaries can be controlled by a basis other than the computational basis: 

\begin{definition} \label{def:control} \textbf{(Basis-controlled unitary)}
A \emph{2-qubit basis-controlled unitary} is a unitary $\CU{A}{B}{C}$ such that, if $\mathcal{A}$ is the basis (for 1 qubit) $\{\ket{a}, \ketp{a}\}$, then
\[
	\CU{A}{B}{C}: \begin{matrix} \quad & \ket{a}\ket{\psi} \mapsto \ket{a}B\ket{\psi} \\ \quad & \ketp{a}\ket{\psi} \mapsto \ketp{a}C\ket{\psi}
	\end{matrix}.
\]
$\CU{A}{B}{C}$ has eigenvectors $\ket{a}\ket{b}, \ket{a}\ketp{b}, \ketp{a}\ket{c}, \ketp{a}\ketp{c}$, where $\mathcal{B} = \{\ket{b},\ketp{b}\}$ and $\mathcal{C} = \{\ket{c}, \ketp{c}\}$ are the eigenbases of $B$ and $C$, respectively. We will draw such a gate as in Figure \ref{fig:CU}. Since these gates will always be 2-qubit gates, we will write $\UC{A}{B}{C}$ to represent a basis-controlled unitary controlled on the second qubit, and acting on the first. 
\end{definition}

\begin{figure}[htbp]
\begin{center}
\includegraphics[scale=1.0]{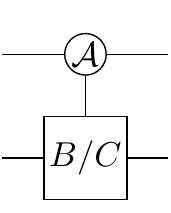}
\caption{A $\CU{A}{B}{C}$ gate.}
\label{fig:CU}
\end{center}
\end{figure}

Some straightforward observations follow. 
\begin{observation} \label{contU}
A continuous version of such a gate can be constructed in the following way. Let $B(t)$ and $C(t)$ be unitary gates such that $B(0)=C(0)=\mathbb{I}$ and $B(T)=B$, $C(T)=C$. Then the continuous version of the basis-controlled unitary is $\CU{A}{B(t)}{C(t)}$, for $0\leq t\leq T$. 
\end{observation}

\begin{observation}
Suppose we have the control-unitary $\CU{A}{B}{B'}$, and $[B,B']=0$. Then
\begin{itemize}
\item $\CU{A}{B}{B'}$ is diagonal in the $\mathcal{A} \otimes \mathcal{B}$ basis.
\item For $B = \unitary{e^{i\theta_1}}{0}{0}{e^{i\theta_2}}$ and $B' = \unitary{e^{i\phi_1}}{0}{0}{e^{i\phi_2}}$, we have $\CU{A}{B}{B'} = \UC{B}{A}{A'}$, where $A = \unitary{e^{i\theta_1}}{0}{0}{e^{i\phi_1}}$ and $A' = \unitary{e^{i\theta_2}}{0}{0}{e^{i\phi_2}}$.
\end{itemize}
\end{observation}

\begin{observation}
Suppose we have the control-unitary $\CU{A}{B}{B'}$, and $B' = e^{i\theta}B$ for some angle $\theta$. Then
\begin{itemize}
\item $\CU{A}{B}{B'} = \CU{A}{e^{i\phi}B}{e^{i\varphi}B} = A \otimes B$ for $A = \unitary{e^{i\theta}}{0}{0}{e^{i\phi}}$ in the $\mathcal{A}$ basis. 
\end{itemize}
\end{observation}

This shows that these $2$ qubit gates include all $1$ qubit gates as a special case. We will now show that all gates on $2$ qubits that have product eigenbases must be in this form. 

\begin{lemma}\label{lem:single_2_qubit}
Suppose we have a 2-qubit unitary with a product eigenbasis. Then this unitary must be a basis-controlled unitary $\CU{A}{B}{C}$ (or $\UC{A}{B}{C}$) for some choice of $\mathcal{A}, B$, and $C$. 
\end{lemma}
\begin{proof}
Since it is product, the eigenbasis must have the form $\{\ket{\psi_0}\ket{\phi_0}, \ket{\psi_1}\ket{\phi_1}, \ket{\psi_2}\ket{\phi_2}, \ket{\psi_3}\ket{\phi_3}\}$. However, as we will now show, the possible choices for the $\psi_i$ and $\phi_i$ are heavily constrained. We can pick an eigenbasis by arbitrarily choosing a basis vector $\ket{a}\ket{c}$, and then choose the subsequent basis vectors under the constraint that they must be orthogonal to all previous basis vectors -- see Figure \ref{fig:bases}.
\begin{figure}[htbp]
\begin{center}
\includegraphics[scale=0.65]{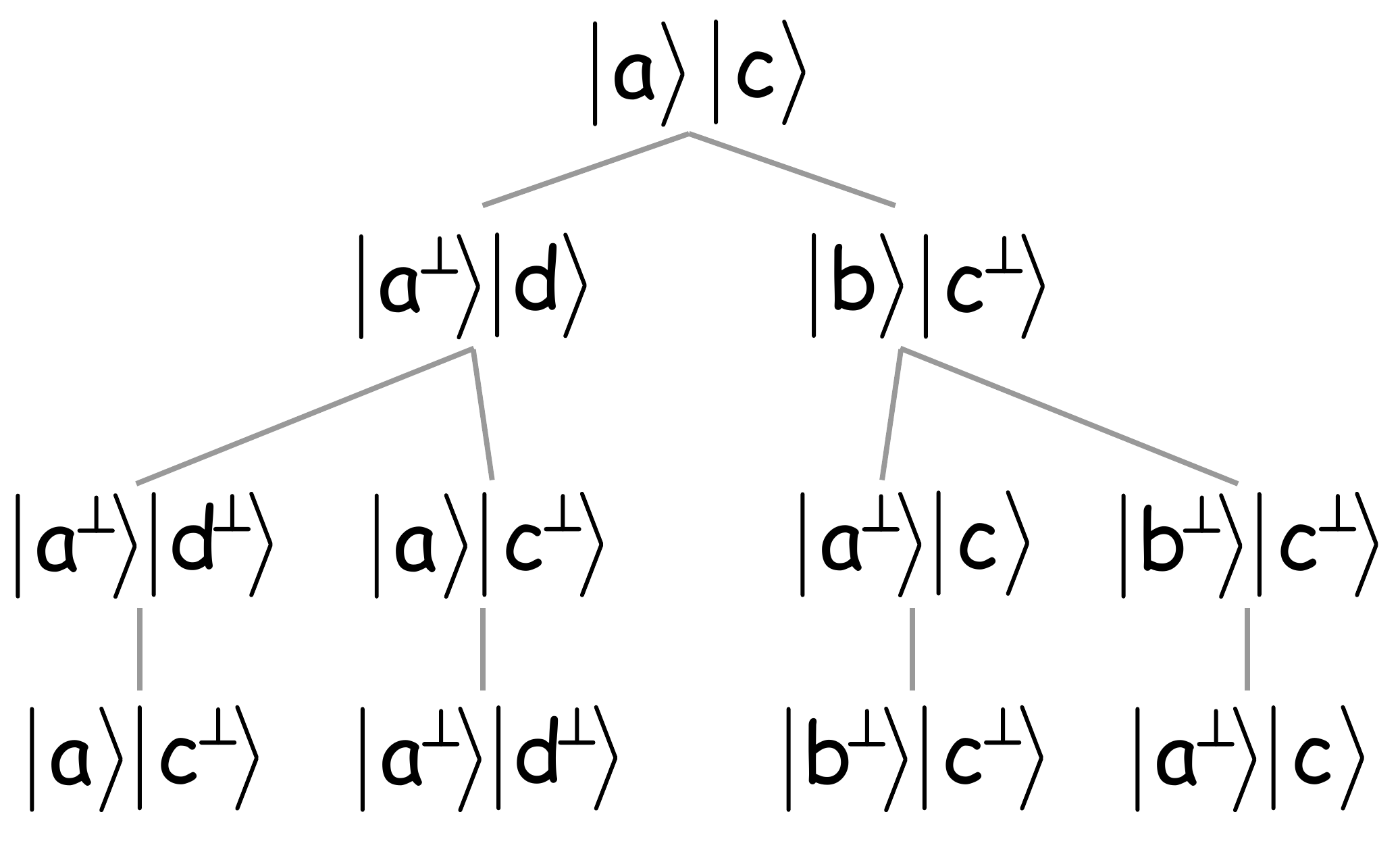}
\caption{All possible choices of product eigenbases for a 2-qubit unitary.}
\label{fig:bases}
\end{center}
\end{figure}
By inspection, one can see that all choices of bases consistent of a single 1-qubit basis on one of the qubits, and two 1-qubit bases on the other qubit, each corresponding to one of the possible basis vectors on the first qubit. Hence, the only basis choice we can have is 
\[
\{ \ket{a}\ket{b}, \ket{a}\ket{b^\perp}, \ket{a^\perp}\ket{c}, \ket{a^\perp}\ket{c^\perp} \},
\]
for some 1-qubit bases $\{\ket{a},\ket{a^\perp}\}, \{\ket{b},\ket{b^\perp}\}, \{\ket{c},\ket{c^\perp}\}$ (up to swapping the first and second qubit). Hence, the unitary $U$ must be of the form 
\[
\proj{a}\left(e^{i\theta_1} \proj{b}+ e^{i\theta_2}\proj{b^\perp}\right) + \proj{a^\perp}\left( e^{i\phi_1} \proj{c}+ e^{i\phi_2}\proj{c^\perp}\right),
\]
where $e^{i\theta_1}, e^{i\theta_2}, e^{i\phi_1},  e^{i\phi_2}$ are the eigenvalues associated with each eigenvector. Written in the $\mathcal{A} = \{\ket{a},\ket{a^\perp}\}$ basis, this is 
\[
\begin{pmatrix}
B & 0 \\
0 & C
\end{pmatrix},
\]
corresponding to a controlled $B/C$ gate: if qubit 1 is in state $\ket{a}$, then the unitary $B = e^{i\theta_1}\proj{b}+e^{i\theta_2}\proj{b^\perp}$ is applied, else if qubit 1 is in state $\ket{a^\perp}$, then unitary $C = e^{i\phi_1}\proj{c}+e^{i\phi_2}\proj{c^\perp}$ is applied. Hence, the unitary must be a basis-controlled unitary $\CU{A}{B}{C}$.
\end{proof}


\subsection{Product control circuits}

In this subsection we will define a type of circuit composed of $1$ and $2$ qubit gates, call a product control circuit, that has a product eigenbasis after each gate (and at each point in time in the continuous version). We will also show that this is the only type of circuit with product eigenbasis, and hence all circuits in $\DQC_{\sep}$ are of this form. 

Every gate in a product control circuit is a basis-controlled unitary of the type in Definition \ref{def:control}. However, these gates cannot be placed arbitrarily. Figure \ref{fig:main_structure} illustrates the possible layouts of the circuits. At every point in the circuit, each qubit can be classified as a control, target, ambiguous or a free qubit. These classification may change as new unitaries are added. We define these classes below. 

\begin{definition}\label{def:types_of_qubit}
We define the following types of qubit in our circuit:
\begin{itemize}
\item Control in $\mathcal{B}$ : A line that is used to control the application of a basis-controlled unitary on a target line, in the basis $\mathcal{B} = \{\ket{b}, \ketp{b}\}$. I.e. the action of the circuit on the target line is either $U_{0,x}$ or $U_{1,x}$, depending on whether the control line is in the state $\ket{b}$ or $\ketp{b}$, and all other control lines are in the state $\ket{x}$. Then there must exist some $x$ such that $[U_{0,x}, U_{1,x}] \neq 0$ (else this qubit would be an ambiguous qubit). The basis of this line is fixed to be $\mathcal{B}$, and it can be acted upon by 2-qubit gates shared with control, target, and ambiguous lines, so long as all these are controlled on this qubit in the $\mathcal{B}$ basis. 
\item Target : Any line that is only acted upon only by control unitaries controlled by a control line.
\item Ambiguous in $\mathcal{B}$ : A line that has a associated basis $\mathcal{B}$, but is not an actual control line. It may apply control unitaries on other lines, however, it is possible to write those (combined) controls as a larger unitary that is diagonal. In the example in Figure \ref{fig:main_structure}, the last gate on the line `ambiguous in $\mathcal{C}$' is diagonal in the basis $\mathcal{A}\otimes \mathcal{B}\otimes...\otimes\mathcal{Z}\otimes\mathcal{C}$. This line is called ambiguous because the next unitary that acts on it can act on it as a control or a target. This would potentially change the classification of the line.
\item Free : A line that is only acted upon by 1-qubit gates.
\end{itemize}
\end{definition}

\begin{figure}[htbp]
\begin{center}
\includegraphics[scale=1.0]{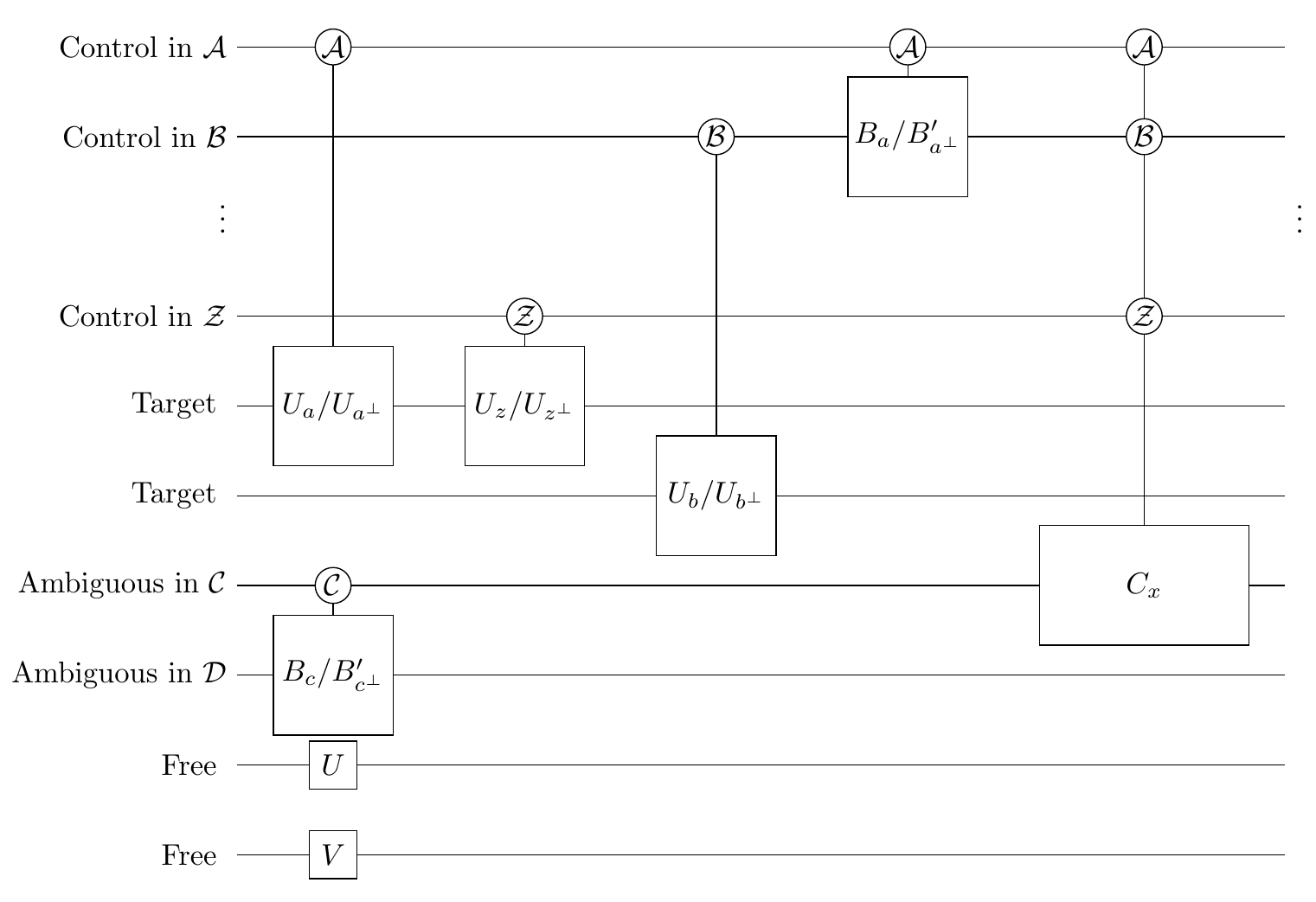}
\caption{This figure gives an example of a circuit that can be constructed using Table \ref{table}. In this example, the first control line applies control unitaries on target and ambiguous lines using $\mathcal{A}$ as the control basis. Gates that are diagonal in $\mathcal{A}\otimes\mathcal{B}$ can be applied between the first two control lines, and similarly, a gate diagonal in $\mathcal{C}\otimes \mathcal{D}$ is possible between ambiguous lines. The final gate illustrated is not a $2$ qubit gate. Instead it is constructed using a series of $2$ qubit gates. In this case it is diagonal in the basis $\mathcal{A}\otimes\mathcal{B}\otimes...\otimes\mathcal{Z}\otimes\mathcal{C}$. See Lemma \ref{lem:lem_3} for more on this kind of gate.}
\label{fig:main_structure}
\end{center}
\end{figure}


\begin{tabular}{ | m{2cm} | m{2cm}| m{3cm} | m{2cm} | m{2cm}| m{2cm} | }
\hline 
$i$ & $j$ & $U$ & Proof & $i$ after & $j$ after\tabularnewline
\hline 
\hline 
Control, $\mathcal{A}$ & Control, $\mathcal{B}$ & diagonal in $\mathcal{A}\otimes\mathcal{B}$ & Corollary \ref{cor:2}  & Control, $\mathcal{A}$  & Control, $\mathcal{B}$\tabularnewline
\hline 
Control, $\mathcal{A}$ & Ambiguous, $\mathcal{B}$ & $\CU{A}{K}{H}$ & Corollary \ref{cor:2} & Control, $\mathcal{A}$  & Ambiguous/ target\tabularnewline
\hline
Control, $\mathcal{A}$ & Free & $\CU{A}{K}{H}$  &  Corollary \ref{cor:2} & Control, $\mathcal{A}$ & Free/ ambiguous/ target\tabularnewline
\hline 
Control, $\mathcal{A}$ & Target (if special case in Lemma \ref{lem:4} doesn't hold) & $\CU{A}{K}{H}$ & Corollary \ref{cor:2} & Control, $\mathcal{A}$ & Target \tabularnewline
\hline 
Control, $\mathcal{A}$ & Target (if special case in Lemma \ref{lem:4} holds) & Either $\CU{A}{K}{H}$, & Corollary \ref{cor:2} & Either control, $\mathcal{A}$, & Target \tabularnewline
& & or $\UC{B}{K}{H}\CU{A}{E}{F}$ &  & or Target/ Ambiguous & Control, $\mathcal{B}$ \tabularnewline
\hline
Target & Target (if cond (i) in Theorem \ref{lem:10} doesn't hold) & None & Theorem \ref{lem:10} & Target  & Target\tabularnewline
\hline 
Target & Target (if cond (i) in Theorem \ref{lem:10} holds) & Either none, & Theorem \ref{lem:10} &  Either target & Target \tabularnewline
 & & or $\CU{A}{K}{H}(W^{\dagger}_x\otimes \mathbb{I}$) & & Or control, $\mathcal{A}$ & Target\tabularnewline
\hline 
Target & Free (if cond (ii) in Theorem \ref{lem:10} doesn't hold) & $\UC{A}{H}{K}$ & Theorem \ref{lem:10}  & Target & Control, $\mathcal{A}$/ free \tabularnewline
\hline
Target & Free (if cond (ii) in Theorem \ref{lem:10} holds) & Either $\UC{A}{H}{K}$, & Theorem \ref{lem:10}  & Target & Control, $\mathcal{A}$/ free \tabularnewline
 & & or $\CU{A}{K}{H}(W^{\dagger}_x\otimes \mathbb{I}$) &  & Control/ ambiguous, $\mathcal{A}$ & Target/ ambiguous \tabularnewline
\hline
\end{tabular}

\begin{table}
\begin{tabular}{ | m{2cm} | m{2cm}| m{3cm} | m{2cm} | m{2cm}| m{2cm} | }
\hline 
$i$ & $j$ & $U$ & Proof & $i$ after & $j$ after\tabularnewline
\hline 
\hline 
Target & Ambiguous, $\mathcal{B}$ (if cond (iii) in Theorem \ref{lem:10} doesn't hold)   & $\UC{B}{K}{H}$ & Theorem \ref{lem:10} & Target/ Ambiguous & Control/ Ambiguous, $\mathcal{B}$ \tabularnewline
\hline 
Target & Ambiguous, $\mathcal{B}$ (if cond (iii) in Theorem \ref{lem:10} holds)  & Either $\UC{B}{K}{H}$, & Theorem \ref{lem:10} & Either target/ ambiguous & Control/ Ambiguous, $\mathcal{B}$ \tabularnewline
&  & Or $\CU{A}{K}{H}(W^{\dagger}_x\otimes \mathbb{I})$ & Theorem \ref{lem:10} & Or control, $\mathcal{A}$ & Target/ Ambiguous, $\mathcal{B}$ \tabularnewline
\hline 
Ambiguous, $\mathcal{A}$ & Ambiguous, $\mathcal{B}$ & Either $\CU{A}{K}{H}$ & Theorem \ref{lem:10} & Either ambiguous/ control, $\mathcal{A}$ & Ambiguous, $\mathcal{B}$/ target \tabularnewline
& & Either $\UC{B}{K}{H}$ & & Or ambiguous, $\mathcal{A}$/ target & Ambiguous/ control, $\mathcal{B}$ \tabularnewline
\hline 
Ambiguous, $\mathcal{A}$ & Free & Either $\CU{A}{K}{H}$ & Theorem \ref{lem:10} & Control/ Ambiguous, $\mathcal{A}$ & Target/ Ambiguous \tabularnewline

& & Or $\UC{B}{K}{H}$ &  & Target/ Ambiguous & Control, $\mathcal{B}$ \tabularnewline
\hline 
Free, $W$ & Free, $V$ & $\CU{A}{K}{H}(W^{\dagger}\otimes V^{\dagger})$ & Theorem \ref{lem:10}  & Control/ Ambiguous, $\mathcal{A}$ & Target/ Ambiguous \tabularnewline
\hline 
\end{tabular}
\caption{\label{table} This table shows what unitaries are allowed between qubit $i$ and $j$, depending on their current classification in the circuit. In this table `Control, $\mathcal{A}$' and `ambiguous, $\mathcal{A}$' mean a control/ ambiguous line with basis $\mathcal{A}$. `Free, V' means a free line which has had the $1$ qubit unitary $V$ applied to it. This table also shows how the line will be classified after the unitary is applied. However this usually depends free parameters in the unitary as well as the previous circuit. These details can be found in the relevant proofs.}
\end{table}

\begin{definition} \textbf{(Product control unitary)}
A $n$ qubit unitary is a product control unitary if it is constructed as a series of $2$ qubit control unitaries in the following way. All qubits are originally classed as `free'. A basis-controlled unitary is applied between $i$ and $j$ according to Table \ref{table}, and the qubits are reclassified accordingly. The process is continued in this way until the full unitary is constructed. A continuous version of this class can be constructed via Observation \ref{contU}.  
\end{definition}

\begin{theorem} \label{thm:prodcontrol}
Any $n$-qubit $\text{DQC1}_{\textrm{sep}}$ circuit is necessarily of the form $\proj{0} \otimes V + \proj{1} \otimes I$ for some $n-1$-qubit unitary $V$, and is such that $V$ is a product control unitary whose corresponding circuit can be constructed uniformly in polynomial time (classically).
\end{theorem}

\begin{proof}
We will prove this claim by induction. Let $k$ be the number of $2$ qubit unitaries in $U$. For $k=1$, Lemma \ref{lem:single_2_qubit} shows that the state is only product if the unitary is a control unitary. This is consistent with Table \ref{table}.

Assume $k=t$ is true. We now need to show that the $k=t+1$ unitary $U$ between $i$ and $j$ is consistent with Table \ref{table}. This is proved both in both Corollary \ref{cor:2} and Theorem \ref{lem:10} in the next section. Some of the special cases listed in Table \ref{table} might be generally computationally hard to check. However, in these cases, there is an alternative $U$ that can be applied regardless of whether that condition holds. Hence the unitaries pertaining to these special cases can only be applied if, in that case, it is possible to check the condition efficiently. 
\end{proof}

\begin{theorem}
Any circuit family from $\DQC_\sep$ can be simulated efficiently classically. That is, $\DQC_\sep \subseteq \BPP$. 
\end{theorem}

\begin{proof}
Suppose a circuit in $\text{DQC1}_{\textrm{sep}}$ is of the form $\proj{0} \otimes V + \proj{1} \otimes I$. We have already seen that the input mixed state can be considered to be the ensemble $\{\frac{1}{2^n},|+\rangle|v_i\rangle \}_i$, where $\{|v_i\rangle\}_i$ is a product eigenbasis of $V$. If it is possible to efficiently sample one of these vectors and compute its eigenvalue then it is possible to simulate the circuit efficiently. 

This is possible in the following manner. If qubit $i$ is classified as a control of ambiguous line with basis $\mathcal{A}=\{|a\rangle,|a^{\perp}\rangle\}$, flip a coin to choose either $|a\rangle_i$ or $|a^{\perp}\rangle_i$. 

Suppose $j$ is a target line and is acted on by control lines so that the $1$ qubit unitary applied to the line is $U_x$, where $x$ labels the state of the control lines. Set $x$ to be the value chosen in the previous step and compute $U_x$, and its eigenvectors $\{|u_x\rangle, |u_x^\perp\rangle$\}. Flip a coin to choose one of these. 

If line $k$ is a free line to which $W$ has been applied, compute the eigenvectors of $W$ and choose one at random. Combining these $1$-qubit eigenvectors gives one of the eigenvectors $|u_i\rangle$ uniformly randomly. To compute its eigenvalue we can evolve the state through the circuit. This can be done efficiently, even though this is not an eigenvector at every time slice. 

\end{proof}


\section{Main proofs} \label{sec:main_proofs}

In this section we prove the results referenced in Theorem \ref{thm:prodcontrol} as well as Theorem \ref{thm:sep}. 

\textbf{Theorem \ref{thm:sep}}
\textit{The final state of a one clean qubit model computation (in which the circuit applied to the qubits is of the form $\proj{0}\otimes \mathbb{I} + \proj{1}\otimes U$) has no entanglement if and only if $U$ has a separable eigenbasis.}

\begin{proof}
The `if' direction is clear. Equation \ref{eq:final} has no entanglement if $\{|u_i\rangle\}_i$ is a separable eigenbasis. The final state is a mixture of these states and so it does not have entanglement in this case.

For the `only if' direction, we will first show that for any ensemble representing the initial state, $\{p_i, |+\rangle|\phi_i\rangle\}_i$, the clean and noisy register become entangled unless each vector $|\phi_i\rangle$ is an eigenvector of $U$ \footnote{The ensemble may be continuous here. This will not change the analysis.}. 

Suppose $\{p_i, |+\rangle|\phi_i\rangle\}_i$ is an ensemble for $\proj{+}\otimes\mathbb{I}/2^n$. Then the pure states in this ensemble evolve in the following way:

\begin{equation}
|+\rangle|\phi_i\rangle \rightarrow \frac{1}{\sqrt{2}} |0\rangle|\phi_i\rangle+\frac{1}{\sqrt{2}} |1\rangle U|\phi_i\rangle. 
\end{equation}

Let $A$ be the clean register, and $B$ the noisy one. We wish to show that this state has no entanglement between $A$ and $B$ if and only if $|\phi_i\rangle$ is an eigenvector of $U$. 

As a density matrix, this state is 
\begin{equation}
\rho_{AB}\equiv \frac{1}{2} |0\rangle\langle 0|\otimes |\phi_i\rangle\langle \phi_i| +\frac{1}{2} |0\rangle\langle 1|\otimes |\phi_i\rangle\langle \phi_i| U^{\dagger} +\frac{1}{2} |1\rangle\langle 0|\otimes U |\phi_i\rangle\langle \phi_i|+\frac{1}{2} U|1\rangle\langle 1|\otimes |\phi_i\rangle\langle \phi_i|U^{\dagger}.
\end{equation}

And the reduced state on $A$ is
\begin{equation}
\rho_{A}= \frac{1}{2} |0\rangle\langle 0| +\frac{1}{2} |0\rangle\langle 1| \langle \phi_i| U^{\dagger}|\phi_i\rangle + \frac{1}{2} |1\rangle\langle 0| \langle \phi_i| U|\phi_i\rangle +\frac{1}{2} |1\rangle\langle 1| .
\end{equation}

This reduced state is pure if and only if $\rho_{AB}$ has no entanglement between $A$ and $B$. A calculation shows Tr$(\rho_A^2)=1/2+1/2|\langle \phi|U|\phi\rangle|^2$, which is $1$ if and only if $|\phi\rangle$ is an eigenstate of $U$.

This means that the only ensembles without entanglement between the clean and noisy qubits are ones in the form $\{p_j, |+\rangle|u_j\rangle\}_i$, where $|u_i\rangle$ is an eigenvector of $U$. Then it is clear that the only ensembles without entanglement for any bipartition are ones for which each $|u_j\rangle$ in the ensemble is product. 

Apply the inverse of control $U$ on this ensemble to show that
\begin{equation}
\label{eq:identity}
\sum_j p_j |u_j\rangle\langle u_j| = \frac{\mathbb{I}}{2^n}
\end{equation}

For this to be true $\{|u_j\rangle\}_j$ must necessarily span the full Hilbert space.
\end{proof}

The following Lemma will be used throughout: 
\begin{lemma}\label{lem:product}
Suppose we have a state on $|C| + |T|$ qubits, $\ket{\psi}_C\ket{\phi}_T = \sum_x \alpha_x \ket{x}_C \ket{\phi}_T$, that is acted upon by a unitary $U = \sum_x \proj{x} \otimes U_x$:
\[
	\sum_x \alpha_x \ket{x}_C \ket{\phi}_T \mapsto U \left(\sum_x \alpha_x \ket{x}_C \ket{\phi}_T\right) = \sum_x \alpha_x \ket{x}_C U_x\ket{\phi}_T.
\]
If the state $\ket{\psi}\ket{\phi}$ is an eigenvector of $U$ with eigenvalue $e^{i\theta}$, then for all $\alpha_x \neq 0$, $\ket{\phi}_T$ is also an eigenvector of $U_x$ with eigenvalue $e^{i\theta}$. 
\end{lemma}
\begin{proof}
We have
\[
	\sum_x \alpha_x \ket{x}_C U_x\ket{\phi}_T = e^{i\theta}\left( \sum_x \alpha_x \ket{x}_C \ket{\phi}_T  \right)
\]
for some phase $\theta$. Write $\ket{\phi} = \sum_y \beta_y \ket{y}$. Then 
\[
\sum_{x,y} \alpha_x \beta_y \ket{x}_C U_x\ket{y}_T = e^{i\theta}\sum_{x,y} \alpha_x \beta_y \ket{x} \ket{y}.
\]
In particular, this means that for every $x$, $y$,
\[
\ket{x}U_x\ket{y} = e^{i\theta}\ket{x}\ket{y}
\]
Since $\{\ket{x}\ket{y}, \ket{x}\ket{y}\}_{x,y}$ forms a orthonormal basis over both registers, this condition can only hold if $U_x\ket{y} = e^{i\theta}\ket{y}$ for all $x,y$.  
\end{proof}


\begin{lemma}\label{lem:single_qubit_gates}
Suppose we have the following circuit: 

\begin{figure}[htbp]
\begin{center}
\includegraphics[scale=1.0]{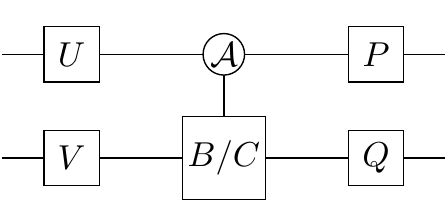}
\label{fig:lemma2}
\end{center}
\end{figure}

Then its corresponding unitary has a product eigenbasis in the following cases:
\begin{itemize}
\item If $[QBV, QCV] \neq 0$, and $U$ and $P$ are diagonal in the $\mathcal{A}$ basis.
\item If $[QBV, QCV] = 0$.
\item If $B = e^{i \phi}C$ for some angle $\phi$.
\end{itemize}
\end{lemma}
\begin{proof}
We begin by absorbing $V$ and $Q$ into the basis-controlled unitary. Let $B' := QBV$ and $C' := QCV$. Assume that $B'$ and $C'$ do not commute, and (wlog) that $\mathcal{A}$ is the computational basis (the same result can be obtained by writing $a$ and $a^\perp$ in place of 0 and 1 in what follows). Choose some eigenvector $\ket{\phi_1}\ket{\phi_2}$ of the circuit, and write $U\ket{\phi_1} = c\ket{0}+d\ket{1}$. Then the action of the circuit on this eigenvector is 
\[
\ket{\phi_1}\ket{\phi_2} = c\ket{0}\ket{\phi_2}+d\ket{1}\ket{\phi_2} \mapsto c\ket{0}B'\ket{\phi_2}+d\ket{1}C'\ket{\phi_2} \mapsto c P\ket{0}B'\ket{\phi_2}+d P\ket{1}C'\ket{\phi_2}.
\]
Since $P$ is unitary, $P\ket{0}$ and $P\ket{1}$ are orthogonal, and so the final state is only product if either $c = 0$ or $d = 0$, or if $B'\ket{\phi_2} \propto C'\ket{\phi_2}$. In the latter case, by Lemma \ref{lem:product}, $\ket{\phi_2}$ must be an eigenvector of both $B'$ and $C'$, which implies that $[B',C']=0$. By assumption, this is not the case, and so we must have that $c=0$ or $d=0$. 
This implies that $U = P^\dag$, and also that $\ket{\phi_2}$ is an eigenvector of at least one of $B'$ and $C'$. 
\\\\
Now we consider the case where $[B',C'] = 0$. In this case, we know from Section \ref{sec:defs} that $\CU{A}{B'}{C'} = \UC{E}{A}{A'}$, where $\mathcal{E}$ is the shared eigenbasis of $B'$ and $C'$ and $A, A'$ are as defined in Section \ref{sec:defs}. Moreover, we can absorb $U$ and $P$ into the basis-controlled unitary, yielding $\UC{E}{PAU}{PA'U}$.
\\\\
Finally, if $B = e^{i\phi}C$ for some angle $\phi$, then from Section \ref{sec:defs} we know that the circuit collapses to two sets of 3 unitaries acting on each qubit separately.
\end{proof}
The following corollary follows immediately from Lemma \ref{lem:single_2_qubit}, and demonstrates how to construct a single 2-qubit basis-controlled unitary in each of the three cases above.
\begin{corollary}
If the circuit from Lemma \ref{lem:single_qubit_gates} has a product eigenbasis, then it can be written as a single basis-controlled unitary. 
\end{corollary}
\begin{proof}
The correctness follows from Lemma \ref{lem:single_2_qubit}. To find the correct form for the basis-controlled unitary, we observe that if $[QCV,QBC]=0$, then we can simply absorb the unitaries $U, V, P, Q$ into $\CU{A}{B}{C}$ as in the second half of the proof of Lemma \ref{lem:single_qubit_gates}, yielding a basis-controlled unitary $\UC{E}{A}{A'}$. If $B = e^{i\phi}C$ for some angle $\phi$, then we can write the entire circuit as two single-qubit unitaries $UAP \otimes QBV$, where $A = \unitary{1}{0}{0}{e^{i\phi}}$. Finally, if $[QCV,QBV]\neq 0$, then $U$ and $W$ must be diagonal in the $\mathcal{A}$ basis, and so they just contribute a relative phase $e^{i\theta}$ which can be absorbed into $B$ and $C$, yielding the basis-controlled unitary $\CU{A}{e^{i\theta}QBV}{e^{i\theta}QCV}$
\end{proof}

Occasionally, we might encounter circuits containing gates controlled on a number of qubits, and acting on a target qubit with a unitary that is diagonal in that qubit's associated basis. In these cases, we can replace the gate with an equivalent gate that is controlled on a different set of qubits, and that acts on a different target qubit, but still with a unitary that is diagonal in that qubit's basis -- see Figure \ref{fig:swap_targets}.

\begin{figure}[htbp]
\begin{center}
\includegraphics[scale=0.8]{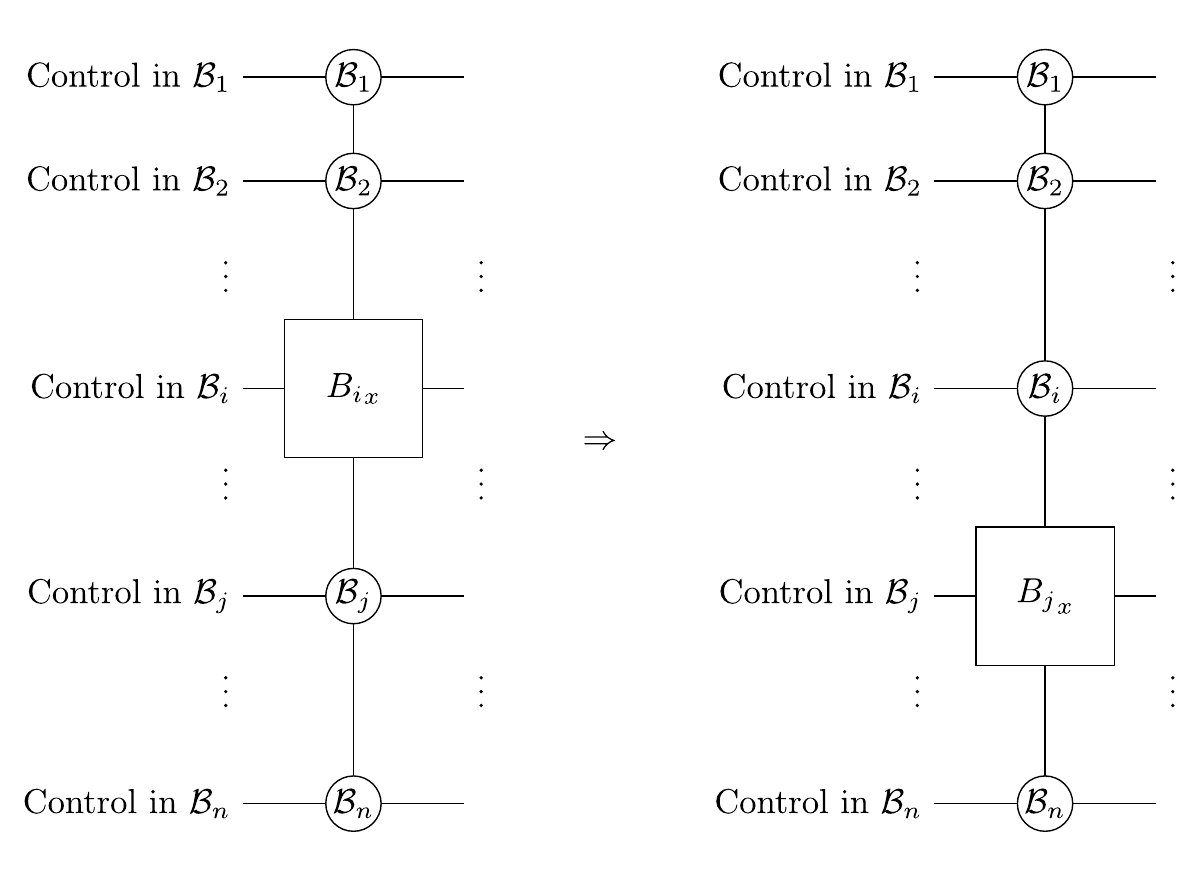}
\caption{Changing the target qubit for a multi-controlled diagonal gate.}
\label{fig:swap_targets}
\end{center}
\end{figure}

\begin{lemma}\label{lem:lem_3}
Suppose that we have an $n$-qubit circuit composed of a single target qubit $i$, with $n-1$ control unitaries acting on it, each controlled on different control lines (which can be controlled in different bases). Let the basis on the $k$th line be $\mathcal{B}_k = \{\ket{b^k_0}, \ket{b^k_1}\}$ (where $\ket{b^k_1}=\ketp{b^k_0}$). Let $x = x_1, x_2, \dots, x_{n-1}$ be an $(n-1)$-bit string, and $U^{(i)}_x$ be the unitary applied to qubit $i$ when the other control lines are in the state $\ket{x} = \ket{b^1_{x_1}} \otimes \ket{b^2_{x_2}} \otimes \cdots \otimes \ket{b^n_{x_{n-1}}}$. 

Further suppose that $U^{(i)}_x$ is diagonal (in the $\mathcal{B}_i$ basis) for all $x$. Then the circuit can be re-written as a basis-controlled unitary acting on any qubit $j \in [n]$, controlled on all the others, such that the unitary $U^{(j)}$ acting on the new target qubit $j$ remains diagonal in the basis $\mathcal{B}_j$ of that qubit.
\end{lemma}
\begin{proof}
Let the eigenvalues of $\ket{b^0_0}\ket{x}$ and $\ketp{b^0_1}\ket{x}$ be $e^{i \theta_{0,x}} = e^{i \theta_{1,x_0,x_1,\dots,x_{n-1}}}$ and $e^{i \theta_{1,x_0,x_1,\dots,x_{n-1}}}$. 

Now let $\ket{y, 0}$ (resp. $\ket{y,1}$) be the state where all lines but the $j$th are in state $\ket{y}$, and the $j$th qubit is in state $\ket{b^j_0}$ (resp $\ket{b^j_1}$). That is, the qubits are in the state $\ket{b^0_{y_0}}\otimes \cdots \otimes \ket{b^j_{b}} \otimes \cdots \otimes \ket{b^n_{y_n}}$ for $b=0$ (resp. $b=1$). Then $B_y$ has eigenvalues $e^{i \theta_{y_1, \dots, 0_j, \dots, y_n}}$ and $e^{i \theta_{y_1, \dots, 1_j, \dots, y_n}}$ corresponding to eigenvectors $\ket{y,0}$ and $\ket{y,1}$, respectively. 

Hence, $B_y$ is diagonal in basis $\mathcal{B}_j$. Furthermore, given that one knows this condition is met, it is efficient to compute $B_y$ for any $y$.
\end{proof}

Whenever we add a 2-qubit gate $U$ to a circuit $C$, all qubits that have an associated basis (i.e. control, ambiguous, or free qubits) and aren't acted upon by this new gate remain unchanged. That is, in the eigenbasis of the new circuit $UC$, the basis associated to each of these qubits remains unchanged from the eigenbasis of the previous circuit $C$. The following lemma formalises this observation.
\begin{lemma}\label{lem:4}
Suppose that we have a circuit $C$ on $n$ qubits in the product control form (see Figure \ref{fig:main_structure}). Now suppose that a 2-qubit gate $U$ is applied to qubits $i$ and $j$, and that the unitary of the resulting circuit has a product eigenbasis $\mathcal{P}$. 

Then for all qubits $k \neq i, j$ corresponding to control, ambiguous, or free lines in basis $\mathcal{B}_k$ (as defined in Definition \ref{def:types_of_qubit}), the $k$th qubit of every eigenvector in $\mathcal{P}$ is in the state $\ket{b^{(k)}_0}$ or $\ket{b^{(k)}_1}$ from $\mathcal{B}_k$.
\end{lemma}
\begin{proof}
Let $S$ be the set of all control, ambiguous, and free lines other than $i$ and $j$, and let $S' = [n]\setminus S$ be the rest. Let $x$ be an $|S|$-bit string, and suppose that all qubits in $S$ are in the state $\ket{x}_S := \bigotimes_{k \in S} \ket{b^{(k)}_{x_k}}$. For some state $\ket{\psi}$ on the qubits in $S'$, we have $\ket{\psi}_{S'}\otimes\ket{x}_S \mapsto^{C} \ket{\psi'}_{S'} \otimes \ket{x}_S$, and we can write the action of $C$ as
\[
C = \sum_{x \in \{0,1\}^{|S|}} (C_{x})_{S'} \otimes \proj{x}_S,
\]
where $(C_{x})_{S'}$ is a unitary acting on the qubits in $S'$. Let $\mathcal{U}\mathcal{C}_x$ be an eigenbasis for $UC_x$. Then $\bigcup_{x \in \{0,1\}^{|S|}} \left(\mathcal{U}\mathcal{C}_x \otimes \proj{x}_S\right)$ is an eigenbasis of $UC$, and is product if the basis $\mathcal{U}\mathcal{C}_x$ is product for all $x \in \{0,1\}^{|S|}$. 
\\\\
Suppose by way of contradiction that there is no such product eigenbasis for $UC_{x'}$, for some $x'$, but that there \emph{is} a product eigenbasis for the whole circuit $UC$. Choose an eigenvector $\ket{\phi}_{S'} \otimes \ket{\psi}_S$ from that basis, and write $\ket{\psi}_S = \sum_{x \in \{0,1\}^{|S|}} \alpha_x \ket{x}_S$. Note that there must exist some choice of $\ket{\phi}$ such that $\alpha_{x'} \neq 0$, else the set of eigenvectors wouldn't span the entire space over the $n$ qubits. Without loss of generality, assume that we have chosen such a $\ket{\phi}$. 

The circuit $UC$ acts on this eigenvector as 
\[
	\ket{\phi}_{S'} \otimes \ket{\psi}_S \mapsto \sum_{x \in \{0,1\}^{|S|}} \alpha_x \left( UC_{x}\ket{\phi}\right) \otimes \ket{x}_S \quad = \quad e^{i\theta} \ket{\phi}_{S'} \otimes \ket{\psi}_S
\]
for some angle $\theta$. 
By Lemma \ref{lem:product}, the equality can only be true if $\ket{\phi}$ is an eigenvector of all $UC_x$ such that $\alpha_x \neq 0$ (which includes $x'$ by assumption). Since this holds for any eigenvector that we choose from the eigenbasis of $UC$, this contradicts our assumption that there exists no product eigenbasis for $UC_x$. Hence there must exist a product eigenbasis for $UC$ whose eigenvectors are of the form $\ket{\phi}_{S'} \otimes \ket{x}_S$ for all $x \in \{0,1\}^{|S|}$. 

In particular, this means that for any qubit $k \neq i, j$ corresponding to a control, ambiguous, or free line in basis $\mathcal{B}$, the $k$th qubit of any eigenvector of the circuit $UC$ remains unchanged from the corresponding eigenvector of the circuit $C$.
\end{proof}

The next lemma, and the corollary that follows, show what constraints must be obeyed by a new 2-qubit gate if it is to act on a control qubit $i$ in basis $\mathcal{A}$ and another qubit $j$. In general, the new unitary must be a control-unitary, controlled on qubit $i$ in basis $\mathcal{A}$ and acting on qubit $j$. However, if a certain condition is met, then it is possible for the gate to be a control-unitary of the form $\CU{E}{B}{B'}$ for some basis $\mathcal{E}$ and commuting unitaries $B$ and $B'$. In this case, by Lemma \ref{lem:4}, the gate can be controlled on either one of $i,j$, whilst acting on the other. 
\begin{lemma}\label{lem:5}
Once again, suppose that we have a circuit $C$ on $n$ qubits in product control form, and that a 2-qubit gate $U$ is applied to qubits $i$ and $j$, and also that the unitary of the resulting circuit has a product eigenbasis.

Suppose that qubit $i$ is a control qubit in basis $\mathcal{A}$ (which for simplicity, and wlog, we will assume is just the computational basis). Then, either
\begin{itemize}
\item[(i)] There exists a product eigenbasis of $UC$ in which line $i$ always has state $\ket{0}_i$ or $\ket{1}_i$, or
\item[(ii)] Line $j$ is the only target line that $i$ acts on, and the following condition holds. Let all other control lines be in the state $\ket{x}$ and suppose that $C$ is such that $V_{0,x}$ is applied to $j$ when $i$ is in state $\ket{0}_i$, and $V_{1,x}$ is applied to $j$ when $i$ is in state $\ket{1}_i$. Then there exist unitaries $D_0$ and $D_1$ such that both $D_0V_{0,x}$ and $D_1V_{1,x}$ are diagonal in a basis $\mathcal{D}$ for all $x$. 
\end{itemize}
\end{lemma}
\begin{proof}
Similarly to the proof of Lemma \ref{lem:4}, we can consider the actions and eigenbases of the circuits $UC_x$ for $x \in \bit^{|S|}$, where $S$ is the set of all control, ambiguous, and free lines other than $i$ and $j$. Let $T$ be the set of all target qubits outside of $i$ and $j$. 
\\\\
First, suppose that there is at least one line $l$ in $T$ that $i$ acts as a control on. Let $W^{(t)}_{0,x}$ (resp. $W^{t}_{1,x}$) be the action of $U$ on qubit $t \in T$ when qubit $i$ is in state $\ket{0}_i$ (resp. $\ket{1}_i$). For an arbitrary (product) eigenvector $\ket{\psi}_i\ket{\phi}_j\ket{x}_{S}\ket{\varphi}_T$, then writing $\ket{\psi}_i = \alpha\ket{0}_i + \beta\ket{1}_i$, the action of $UC_x$ on this state is
\begin{eqnarray*}
	(\alpha\ket{0}_i + \beta\ket{1}_i)\ket{\phi}_j\ket{x}_S\ket{\varphi}_T 
	\mapsto 
	&\alpha U(\ket{0}_i V_{0,x} \ket{\phi}_j) \ket{x}_{S} W^{(1)}_{0,x} \otimes \cdots \otimes W^{(|T|)}_{0,x} \ket{\varphi}_T + \\
	&\beta U(\ket{1}_i V_{1,x} \ket{\phi}_j) \ket{x}_{S} W^{(1)}_{1,x} \otimes \cdots \otimes W^{(|T|)}_{1,x} \ket{\varphi}_T,
\end{eqnarray*}
Therefore, given that $\ket{\psi}_i\ket{\phi}_j\ket{x}_{S}\ket{\varphi}_T$ is an eigenvector of $UC_c$, then we must have 
\begin{eqnarray*}
&\alpha U(\ket{0}_i V_{0,x} \ket{\phi}_j) W^{(1)}_{0,x} \otimes \cdots \otimes W^{(|T|)}_{0,x} \ket{\varphi}_T& + \\
&\beta U(\ket{1}_i V_{1,x} \ket{\phi}_j) W^{(1)}_{1,x} \otimes \cdots \otimes W^{(|T|)}_{1,x} \ket{\varphi}_T&
= e^{i\theta}\ket{\psi}_i\ket{\phi}_j\ket{\varphi}_T.
\end{eqnarray*}
for some angle $\theta$. By Lemma \ref {lem:product}, this implies that $\ket{\phi}_T$ is an eigenvector of both $W^{(1)}_{0,x} \otimes \cdots \otimes W^{(|T|)}_{0,x}$ and $W^{(1)}_{1,x} \otimes \cdots \otimes W^{(|T|)}_{1,x}$ with eigenvalue $e^{i\theta}$. Since, by assumption, $\ket{\phi}$ is a product state $\ket{\phi^{(1)}} \otimes \ket{\phi^{(2)}} \otimes \cdots \otimes \ket{\phi^{(l)}} \otimes \cdots \otimes \ket{\phi^{(|T|)}}$, then $\ket{\phi^{(l)}}$ must be an eigenvector of both $W^{(l)}_{0,x}$ and $W^{(l)}_{1,x}$. This means that qubit $i$ does not act (non-trivially) as a control qubit on qubit $l$, contradicting our assumption above. Hence, there cannot be a qubit in $T$ that is acted upon as a control by $i$. 
\\\\
Henceforth we will assume that the only qubit that $i$ acts as a control on is $j$. Then the action of the circuit on the same eigenvector (omitting the $S$ (control) register) is 
\[
	\alpha U(\ket{0}_i V_{0,x} \ket{\phi}_j) \ket{\varphi}_T + \beta U(\ket{1}_i V_{1,x} \ket{\phi}_j) \ket{\varphi}_T   \quad \propto \quad   \ket{\psi}_i\ket{\phi}_j\ket{\varphi}_T.
\]
Assume $\alpha, \beta \neq 0$. Then because $U(\ket{0}_iV_{0,x}\ket{\phi}_j)$ and $U(\ket{1}_iV_{1,x}\ket{\phi}_j)$ are orthogonal, we can write $U(\ket{0}_iV_{0,x}\ket{\phi}_j) = \ket{u}_i\ket{\phi}_j$ and $U(\ket{1}_iV_{1,x}\ket{\phi})_j) = \ketp{u}_i\ket{\phi}_j$, for some state $\ket{u}$. But because $\{\ket{0}\ket{\phi}, \ket{1}\ket{\phi}, \ket{0}\ketp{\phi}, \ket{1}\ketp{\phi} \}$ forms an orthonormal basis for the two qubits, this implies that 
\[
	\ket{0}_i\ketp{\phi}_j \mapsto U(\ket{0}_iV_{0,x}\ketp{\phi}_j) = \ket{\tilde{u}}_i\ketp{\phi}_j
\]
\[
	\ket{1}_i\ketp{\phi}_j \mapsto U(\ket{1}_iV_{1,x}\ketp{\phi}_j) = \ketp{\tilde{u}}_i\ketp{\phi}_j
\]
for some state $\ket{\tilde{u}}$. 
\\\\
Define $E_x$ and $F_x$ by their actions on basis states: 
\[
\begin{matrix}
E_x\ket{0} = \ket{u} & \quad & F_x\ket{0} = \ket{\tilde{u}} \\
E_x\ket{1} = \ketp{u} & \quad & F_x\ket{1} = \ketp{\tilde{u}} 
\end{matrix},
\]
and let $\mathcal{D}$ be the basis $\{\ket{\phi}, \ketp{\phi}\}$. Then $U(\CU{A}{V_{0,x}}{V_{1,x}}) = \UC{D}{E_x}{F_x}$, implying that we can decompose $U$ as 
\[
	U = \CU{A}{V^\dag_{0,x}}{V^\dag_{1,x}} \UC{D}{E_x}{F_x},
\]
(See Figure \ref{fig:lem_5_proof} for clarity.) 
\begin{figure}[htbp]
\begin{center}
\includegraphics[scale=1.0]{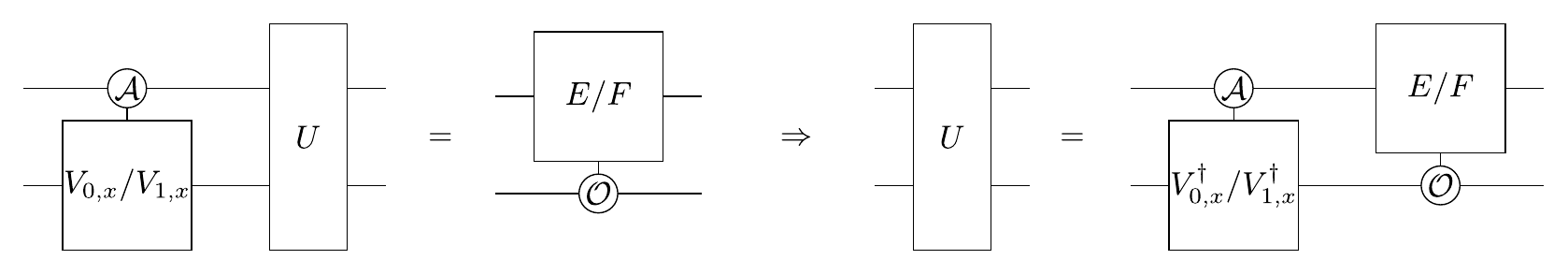}
\caption{$U(\CU{A}{V_{0,x}}{V_{1,x}}) = \CU{O}{E}{F}$ implies that $U = \CU{A}{V^\dag_{0,x}}{V^\dag_{1,x}} \CU{O}{E}{F}$.}
\label{fig:lem_5_proof}
\end{center}
\end{figure}

We require the circuit $UC_x$ to have a product eigenbasis for all $x$. In particular, the circuit formed by applying $U$ after $\CU{A}{V_{0,x'}}{V_{1,x'}}$ must have a product eigenbasis. Restricting our attention to qubits $i$ and $j$, the unitary acting on these two qubits must be a basis-controlled unitary (by Lemma \ref{lem:single_2_qubit}). As we have just shown, this circuit can be written as  $\CU{A}{V^\dag_{0,x}V_{0,x'}}{V^\dag_{1,x}V_{1,x'}} \UC{D}{E_x}{F_x}$.
\\\\
In the case $[E_x,F_x]\neq 0$, then both $V^\dag_{0,x}V_{0,x'}$ and $V^\dag_{1,x}V_{1,x'}$ must be diagonal in the $\mathcal{D}$ basis for all $x,x'$, allowing us to re-write the circuit as a single control unitary, which can be controlled on either of the 2 qubits by Lemma \ref{lem:lem_3}.
\\\\
In the case that $[E,F]=0$, they must be diagonal in the same basis. Supposing that this basis is $\mathcal{B}$, then the circuit can be written as $\CU{A}{V^\dag_{0,x}V_{0,x'}}{V^\dag_{1,x}V_{1,x'}} \CU{B}{D}{D'}$ for some $D, D'$ such that $[D,D']=0$, which is only a valid control-unitary (which indeed it must be, by Lemma \ref{lem:single_2_qubit}) if $\mathcal{A} = \mathcal{B}$. This circuit has an eigenbasis with the $i$th qubit always set to $\ket{0}_i$ or $\ket{1}_i$ (since the basis of qubit $i$ must be $\mathcal{A}$). 
\\\\
In summary, if $i$ is a control line, then $j$ can be its only target, and applying a unitary $U$ to $i$ and $j$ either yields a unitary with an eigenbasis in which the $i$th line is always $\ket{0}_i$ or $\ket{1}_i$, or else $V^\dag_{0,x}V_{0,x'}$ and $V^\dag_{1,x}V_{1,x'}$ are diagonal in basis $\mathcal{D}$ for all $x,x'$, and the unitary $U$ is of the form $\CU{A}{V^\dag_{0,x}V_{0,x'}}{V^\dag_{1,x}V_{1,x'}} \CU{B}{D}{D'}$ for some fixed $x$ and commuting unitaries $D, D'$. Choosing $D_0 = V^\dag_{0,x}$ and $D_1 = V^\dag_{1,x}$ suffices to prove the result. 

\end{proof}

\begin{corollary}\label{cor:2}
Except in case (ii) of Lemma \ref{lem:5}, the unitary $U$ applied to the circuit must be a basis-controlled unitary, controlled by qubit $i$ and acting on qubit $j$. 
\end{corollary}
\begin{proof}
As before, consider the eigenbasis of $UC_x$ for some $x$. We know that in the case we are considering, any eigenvector from this basis has either $\ket{0}_i$ or $\ket{1}_i$ on the $i$th line. Without loss of generality, choose an eigenvector with a $\ket{0}$ on this qubit: $\ket{0}_i\ket{\psi}_j\ket{x}_S\ket{\phi}_T$. The action of the circuit on this state is 
\begin{eqnarray*}
	\ket{0}_i\ket{\psi}_j\ket{x}_S\ket{\phi}_T &\mapsto& \ket{0}_i V_{0,x}\ket{\psi}_j \ket{x}_S W_{0,x}\ket{\phi}_T \\
	&\mapsto& U(\ket{0}_i V_{0,x}\ket{\psi}_j) \ket{x}_S W_{0,x}\ket{\phi}_T \\
	&=& e^{i\varphi} \ket{0}_i\ket{\psi}_j\ket{x}_S\ket{\phi}_T
\end{eqnarray*}
for some angle $\varphi$.
This implies that $\ket{\phi}_T$ is an eigenvector of $W_{0,x}$. Moreover, $\ket{0}_i\ketp{\psi}_j\ket{x}_S\ket{\phi}_T$ must also be an eigenvector of $UC_x$, as must $\ket{1}_i\ket{\theta}_j\ket{x}_S\ket{\phi}_T$ and $\ket{1}_i\ketp{\theta}_j\ket{x}_S\ket{\psi}_T$, for some $\ket{\theta}$. This implies that $U = \proj{0}_i \otimes B + \proj{1}_i \otimes C = \CU{A}{B}{C}$ , with $\ket{\psi}, \ketp{\psi}$ eigenvectors of $BV_{0,x}$ and $\ket{\theta}, \ketp{\theta}$ eigenvectors of $CV_{1,x}$.
\end{proof}

The results so far can be summed up by the following theorem. 
\begin{theorem}
Applying a 2-qubit gate $U$ to qubits $i$ and $j$ after applying the circuit $C$, whilst maintaining that the final circuit $UC$ has a product eigenbasis, places the following constraints on $U$:

\begin{table}[h]
\hspace{-1.5cm}
\begin{tabular}{llll}
\multicolumn{1}{c}{$i$}          & \multicolumn{1}{c}{$j$}                              & \multicolumn{1}{c}{$U$}                       & \multicolumn{1}{c}{Proof}                                                         \\ \hline
Control with basis $\mathcal{A}$ & Control with basis $\mathcal{B}$                     & Diagonal in $\mathcal{A} \otimes \mathcal{B}$ & Corollary \ref{cor:2}                                                                                 \\
Control with basis $\mathcal{A}$ & Ambiguous (with basis $\mathcal{B}$)                 & basis-controlled unitary in basis $\mathcal{A}$        & Corollary \ref{cor:2}                                                                                 \\
Control with basis $\mathcal{A}$ & Free                                                 & basis-controlled unitary in basis $\mathcal{A}$        & \begin{tabular}[c]{@{}l@{}} Corollary \ref{cor:2}\end{tabular} \\
Control with basis $\mathcal{A}$ & Target (Except in special case of Lemma \ref{lem:5}) & basis-controlled unitary in basis $\mathcal{A}$        & \begin{tabular}[c]{@{}l@{}} Corollary \ref{cor:2}\end{tabular}
\end{tabular}
\end{table}
\end{theorem}


We now deal with the constraints placed on gates that act on ambiguous qubits. 

\begin{lemma}\label{lem:7}
Suppose that the circuit $C$ has already been applied, and now a 2-qubit gate $U$ is applied to qubits $i$ and $j$, which are both ambiguous in bases $\mathcal{A}$ and $\mathcal{B}$, respectively. Then in order for $UC$ to have a product eigenbasis, $U$ must be a basis-controlled unitary controlled on either $i$, in basis $\mathcal{A}$; $j$, in basis $\mathcal{B}$; or both. 
\end{lemma}
\begin{proof}
Once again, we will consider the action of the circuit $UC_x$ for some $x$. We know that so far all gates will have acted on line $i$ in basis $\mathcal{A}$ and line $j$ in basis $\mathcal{B}$. There may also have been some unitaries acting between them of the form $\CU{A}{B}{B'}$ (for commuting $B, B'$), which we can combine into a single gate $P_{{\mathcal{A}\mathcal{B}}_x}$ diagonal in $\mathcal{A} \otimes \mathcal{B}$. The resulting circuit $UC_x$, when restricting attention to qubits $i$ and $j$, is $U P_{{\mathcal{A}\mathcal{B}}_x}$. Since this circuit is promised to have a product eigenbasis, it must be of control-unitary form, which (wlog) we will write as $\CU{E}{H_x}{K_x}$. This implies that 
$U = \CU{E}{H_x}{K_x} {{P_{\mathcal{A}\mathcal{B}}}_x}^\dag$ (see Figure \ref{fig:ambiguous_full} for clarity).

\begin{figure}
\hspace{-2cm}
    \includegraphics[scale=1.0]{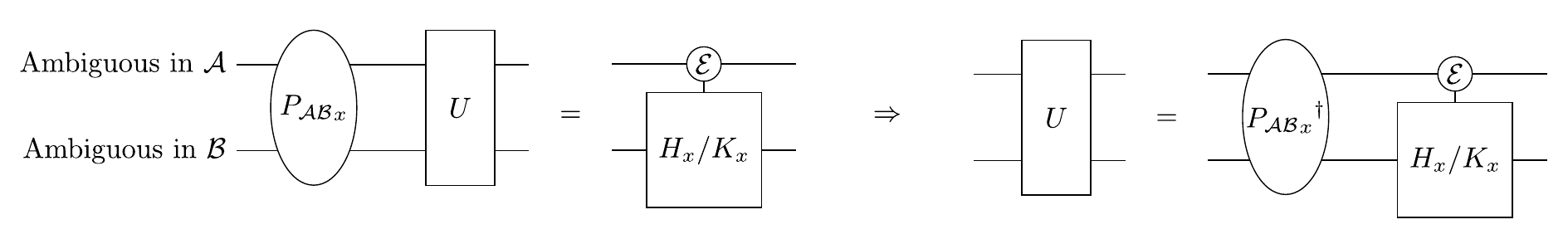}
    \caption{$U P_{{\mathcal{A}\mathcal{B}}_x} = \CU{E}{H_x}{K_x}$ implies that $U = \CU{E}{H_x}{K_x} {{P_{\mathcal{A}\mathcal{B}}}_x}^\dag$}    \label{fig:ambiguous_full}
\end{figure}

But since the condition must hold for other values of $x$, this implies that for all $x' \neq x$, 
$UC_{x'} = \CU{E}{H_x}{K_x} {{P_{\mathcal{A}\mathcal{B}}}_x}^\dag {{P_{\mathcal{A}\mathcal{B}}}_{x'}}^\dag.$ (shown in Figure \ref{fig:ambiguous_4} for clarity).
If $P_{{\mathcal{A}\mathcal{B}}_x} = P_{{\mathcal{A}\mathcal{B}}_{x'}}$ for all $x,x'$, then these cancel, and we have $U = \CU{E}{H_x}{K_x}$. 

\begin{figure}
\centering
    \includegraphics[scale=1.0]{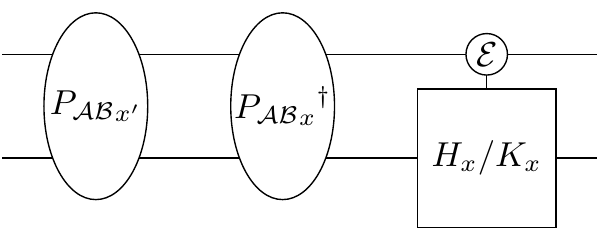}
    \caption{$\CU{E}{H_x}{K_x} {{P_{\mathcal{A}\mathcal{B}}}_x}^\dag {{P_{\mathcal{A}\mathcal{B}}}_{x'}}^\dag.$}    \label{fig:ambiguous_4}
\end{figure}

Otherwise, the two unitaries $P_{{\mathcal{A}\mathcal{B}}_{x}}$ and $P_{{\mathcal{A}\mathcal{B}}_{x'}}$ will combine to form some circuit $\CU{E}{H_x}{K_x} {P_{{\mathcal{A}\mathcal{B}}_{x}}}^\dag P_{{\mathcal{A}\mathcal{B}}_{x'}} =: \CU{E}{H_x}{K_x} P_{{\mathcal{A}\mathcal{B}}_{xx'}}$, which must have a product eigenbasis and therefore be a control-unitary itself (by Lemma \ref{lem:single_2_qubit}). We can write the unitary $P_{{\mathcal{A}\mathcal{B}}_{xx'}}$ in control-unitary form as $\CU{A}{P_{{\mathcal{A}\mathcal{B}}_{xx'}}^{(a)}}{P_{{\mathcal{A}\mathcal{B}}_{xx'}}^{(a^\perp)}}$, where ${P_{{\mathcal{A}\mathcal{B}}_{xx'}}^{(b)}}$ is the action of $P_{{\mathcal{A}\mathcal{B}}_{xx'}}$ on qubit $j$ when qubit $i$ is in state $\ket{b} \in \{\ket{a},\ketp{a}\}$. For this circuit to be a valid basis-controlled unitary, then either $\mathcal{A} = \mathcal{E}$, or $H_x$ and $K_x$ are diagonal in $\mathcal{B}$, or both. 

In the first case, for some $x'$, the resulting circuit must be a control-unitary of the form $\CU{A}{U^{(a)}_{x'}}{U^{(a^\perp)}_{x'}}$. In the second case, where $H_x$ and $K_x$ are diagonal in $\mathcal{B}$, then the circuit is of the form $\UC{B}{U^{(b)}_{x'}}{U^{(b^\perp)}_{x'}}$, where in both cases we define $U^{(a)}_{x'} = H_x P_{{\mathcal{A}\mathcal{B}}_{xx'}}$, and the others are defined similarly. If both $H_x$ and $K_x$ are diagonal in $\mathcal{B}$, and $\mathcal{E}=\mathcal{A}$, then the circuit is of the form $\CU{A}{B}{B'}$ for commuting $B, B'$, and the qubits $i$ and $j$ remain ambiguous in bases $\mathcal{A}$ and $\mathcal{B}$, respectively.

If, instead, only one of the conditions holds, then if $[U^{(a)}_{x}, U^{(a^\perp)}_{x}] \neq 0$ for all $x$, then the control line (which is either $i$ or $j$ above) becomes a true control line in basis $\mathcal{A}$ (as per Definition \ref{def:types_of_qubit}), and the target line becomes a true target line. Otherwise, qubits $i$ and $j$ remain ambiguous in their respective bases. 
\end{proof}

The remaining cases can be summarised by the following theorem.
\begin{theorem}\label{lem:10}
Applying a 2-qubit gate $U$ to qubits $i$ and $j$, after applying the circuit $C$, whilst maintaining that the final circuit $UC$ has a product eigenbasis places the following constraints on $U$:

\begin{table}[h]
\begin{tabular}{c|c|c}
$i$    & $j$                        & $U$                                                                                                                                           \\ \hline
Target & Target                     & \begin{tabular}[c]{@{}c@{}}Does not exist unless \\ condition (i) below is satisfied\end{tabular}                                               \\
Target & Free                       & \begin{tabular}[c]{@{}c@{}}Control on $j$ in any basis\\ (Or control on $i$ if condition (ii) below is satisfied)\end{tabular}                   \\
Target & Ambiguous in $\mathcal{B}$ & \begin{tabular}[c]{@{}c@{}}Control on $j$ in basis $\mathcal{B}$\\ (Or a control in any basis if condition (iii) below is satisfied)\end{tabular} \\
Free   & Free                       & Control on either $i$ or $j$, in any basis.                                                                                                  
\end{tabular}
\end{table}

\end{theorem}
\begin{proof}
In all of the above cases, we can assume that there have been no unitaries acting between qubits $i$ and $j$ so far (else, as we will see, none of the above cases will exist). Therefore we can view the circuit $C_x$ so far as acting with a unitary $W_x$ on qubit $i$ and a unitary $V_x$ on qubit $j$ (in some cases the dependence on $x$ will be redundant). 

Similarly to before, we can make the following set of claims. The (2-qubit) unitary acting on qubits $i$ and $j$ after applying the new gate $U$ must have a product eigenbasis, and therefore be a basis-controlled unitary (by Lemma \ref{lem:single_2_qubit}). Writing this basis-controlled unitary as $\CU{A}{B_x}{C_x}$, we must have that $U (W_x \otimes V_x) = \CU{A}{B_x}{C_x}$, and therefore $U = \CU{A}{B_x}{C_x} (W^\dag_x \otimes V^\dag_x)$. Once again, the product eigenbasis condition must hold for other values of $x$. In particular, for any $x'$, we must have that $\CU{A}{B}{C} (W^\dag_x \otimes V^\dag_x) (W_{x'} \otimes V_{x'})$ is a basis-controlled unitary. 

If $[B_x, C_x] \neq 0$, then Lemma \ref{lem:single_qubit_gates} implies that $W^\dag_x W_{x'}$ is diagonal in basis $\mathcal{A}$ for all $x, x'$. If $[B_x, C_x] = 0$ for all $x, x'$ (and hence share an eigenbasis $\mathcal{B}$), then either $W^\dag_x W_{x'}$ is diagonal in basis $\mathcal{A}$, or $V^\dag_x V_{x'}$ is diagonal in basis $\mathcal{B}$. If both unitaries are diagonal in their respective bases, then $UC_x$ acts as two single-qubit gates on both $i$ and $j$. 
\\\\
Now we deal with the specific cases stated in the theorem. 
\paragraph{Case 1 -- target : target} As we just saw, $UC$ can only have a product eigenbasis if, for all $x, x'$, either $W^\dag_x W_{x'}$ or $V^\dag_x V_{x'}$ is diagonal in some basis $\mathcal{A}$ or $\mathcal{B}$, respectively. This is condition (i) referenced in the statement of the theorem. In these cases, $U$ can act as a basis-controlled unitary controlled on the basis $\mathcal{A}$ or $\mathcal{B}$. 

\paragraph{Case 2 -- target : free} The action of the circuit $C$ can be written as a tensor product of one unitary that depends on the other control lines in some way, $W_x$, and a single qubit unitary that depends on no other qubits, $V$. I.e. $C = W_x \otimes V$. Then $UC$ can have a product eigenbasis if $U$ acts as $(I\otimes V)\UC{E}{A}{B}$ for any basis $\mathcal{E}$ and 2-qubit unitaries $A, B$. 

If condition (ii) is satisfied -- $W^\dag_x W_{x'}$ is diagonal in some basis $\mathcal{A}$ for all $x, x'$ -- then $U$ may also act as a basis-controlled unitary $\CU{A}{B}{C}$ for any 2-qubit unitaries $B, C$.  

\paragraph{Case 3 -- target : ambiguous} In this case, line $j$ has an associated basis $\mathcal{B}$. $UC$ has a product eigenbasis if $U$ acts as a basis-controlled unitary controlled on basis $\mathcal{B}$ (i.e. is of the form $\UC{B}{A}{B}$ for any 2-qubit unitaries $A,B$). 

If condition (iii) is satisfied -- $W^\dag_x W_{x'}$ is diagonal in basis $\mathcal{A}$ for all $x, x'$ -- then $U$ can act as a basis-controlled unitary controlled on basis $\mathcal{A}$ (i.e. is of the form $\CU{A}{B}{C}$ for any 2-qubit unitaries $B,C$). 

\paragraph{Case 4 -- free : free} In this case, the circuit so far has acted as $C = W \otimes V$, for $W, V$ 2-qubit unitaries independent of the state of the other qubits. Then $UC$ has a product eigenbasis if $U$ acts as $(W^\dag \otimes V^\dag)$ followed by a basis-controlled unitary in any basis, controlled on either $i$ or $j$ (or both if it acts as a basis-controlled unitary from $\mathcal{P}$ or $\mathcal{S}$). 
\end{proof}


\section{Further remarks and future work}

In this work we have shown that the One Clean Qubit model without entanglement is classically simulable. This leaves open the larger question: are all (mixed state) quantum computers classically simulable without entanglement? We conjecture the following. 
\begin{conjecture}
Every uniformly constructed family of circuits without entanglement can be classically simulated. 
That is, for any $n$-qubit circuit $U = U_M\dots U_1$ composed of $M$ elementary gates such that the state $U_t\dots U_1 \ket{0^n}$ is separable for all $1\leq t\leq M$, it is possible to estimate the probability of measuring 0 on the first qubit classically in polynomial time, up to additive accuracy $1/\poly(n)$. 
\end{conjecture}

One method for disproving this conjecture (or rather, showing it to be very unlikely to be true) is to find a class of separable computations for which a `quantum supremacy' result can be proved. Such a result would state that no classical simulation (to multiplicative \cite{bremner2010classical} or additive \cite{aaronson2011computational,bremner2016average}) error can exist unless certain complexity theoretic conjectures are false. Multiplicative error results of this type usually rely on showing that post selection boosts the class to $\PBQP$\footnote{This is not the only method, see Ref \cite{morimae2014hardness,fujii2015power}.}. We conjecture this is not possible for a seperable class of circuits. It does not appear possible to create entangled states from separable ones, which seems necessary to achieve the power of $\PBQP$. Let $\mathsf{BQP}_\sep$ be the class consisting of problems solvable by (mixed state) quantum circuits without entanglement. Our conjecture is that
\begin{conjecture}
$\mathsf{PostBQP}_\sep = \PostBPP$.
\end{conjecture}
Our results imply that $\PDQC_\sep \subseteq \PostBPP$, and suggest that adding post-selection to a quantum computer lacking entanglement cannot increase its power beyond $\PostBPP$. However, even showing the (much) weaker containment $\BQP_\sep \subseteq \PostBPP$ appears to be difficult, for much the same reason that showing $\BQP_\sep \subseteq \BPP$ appears to be difficult.

\section*{Acknowledgements}

The authors would like to thank Florian Venn for coming up with Example \ref{flos}, which convinced us this problem was tractable after all. We also acknowledge Ashley Montanaro and Richard Jozsa for their helpful comments.

\bibliography{bibliography.bib} 

\end{document}